\documentclass[a4paper]{scrartcl}               %
\usepackage{ppaper}                             %

\usepackage{relsize}
\usepackage{amssymb, amsmath,verbatim}
\usepackage[normalem]{ulem} %
\usepackage{siunitx}
\usepackage{booktabs}

\usepackage{enumitem}
\setitemize{noitemsep,topsep=0pt,parsep=0pt,partopsep=0pt}

\newcommand{\ie}{{i.e.}}

\newcommand{\ea}{et\ al.}

\newcommand{\CppPlus}{\protect\hspace{-.1em}\protect\raisebox{.35ex}{\smaller{\smaller\textbf{+}}}}
\newcommand{\Cpp}{\mbox{C\CppPlus\CppPlus}}

\renewcommand{\epsilon}{\varepsilon}
\newcommand{\bigO}{O}
\newcommand{\Poly}{\ensuremath{\mathcal{P}}}
\newcommand{\NP}{\ensuremath{\mathcal{NP}}}

\newcommand{\N}{\mathbb{N}}
\newcommand{\R}{\mathbb{R}}
\newcommand{\Rplus}{\mathbb{R}_{\mathsmaller{\geq 0}}}

\newcommand{\OPT}{\textsl{\small OPT}}

\newcommand{\WRC}{\textsf{WRC}}
\newcommand{\Polygon}{P}
\newcommand{\Edges}{E}
\newcommand{\Edge}[2]{\overline{#1 #2}}

\newcommand{\Holes}{\mathcal{H}}

\newcommand{\Interior}{\mathcal{I}}
\newcommand{\Area}{\mathcal{A}}
\newcommand{\Rectangles}{\mathcal{R}}
\newcommand{\Cover}{\mathcal{C}}
\newcommand{\AllEdges}{\mathcal{E}}
\newcommand{\AllVertices}{\mathcal{V}}
\newcommand{\Rect}{R}
\newcommand{\BaseRect}{B}
\newcommand{\BBox}{\mathbb{B}}
\newcommand{\BaseRects}{\mathcal{B}}
\newcommand{\GridRects}{\mathcal{D}}
\newcommand{\MaxRects}{\mathcal{M}}
\newcommand{\RectPow}{\Gamma}
\newcommand{\BaseRectGraph}{G^*}
\newcommand{\TopLeft}{\tikz\draw[thick] (0,0) |- +(1.3ex,1.3ex);}
\newcommand{\BottomRight}{\tikz\draw[thick] (0,0) -| +(1.3ex,1.3ex);}
\newcommand{\Cost}{c}

\newcommand{\InstanceSet}[1]{\textsf{\textbf{#1}}}
\newcommand{\AlgName}[1]{\textsf{#1}}

\usepackage{graphicx}

\usepackage{wrapfig}
\newcommand{\erclogowrapped}[1]{%
\setlength\intextsep{0pt}%
\begin{wrapfigure}[3]{r}{#1}%
\includegraphics[width=#1]{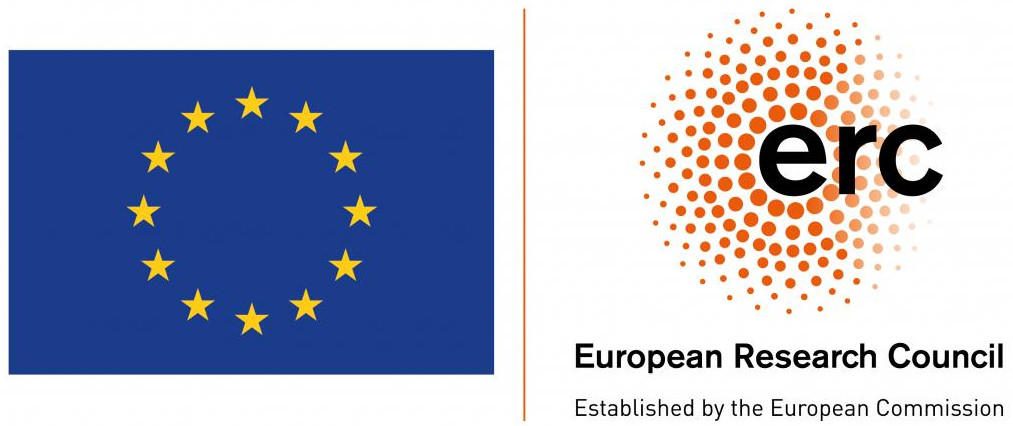}%
\end{wrapfigure}%
}

\usepackage{hyperref} %
\hypersetup{
        colorlinks,
        linkcolor={green!75!black},%
        citecolor={blue!75!black},%
        urlcolor={blue!85!black}
}

\usepackage{tikz}
\usetikzlibrary{patterns}

\title{Covering Rectilinear Polygons\\with Area-Weighted Rectangles}

\author{Kathrin Hanauer}                                                         %
\affiliation{University of Vienna, Faculty of Computer Science, Vienna, Austria} %
\mail{kathrin.hanauer@univie.ac.at}                                              %
\orcid{0000-0002-5945-837X}                                                      %
\author{Martin P. Seybold}                                                       %
\affiliation{University of Vienna, Faculty of Computer Science, Vienna, Austria} %
\mail{martin.seybold@univie.ac.at}                                               %
\orcid{0000-0001-6901-3035}                                                      %
\author{Julian Unterweger}                                                       %
\affiliation{University of Vienna, Faculty of Computer Science, Vienna, Austria} %
\mail{julian.unterweger@univie.ac.at}                                            %
\orcid{0009-0007-0990-6369}                                                      %
\funding{%
\erclogowrapped{7\baselineskip}%
This project has received funding from the European Research Council (ERC)       %
under the European Union's Horizon 2020 research and innovation programme        %
(Grant agreement No.\ 101019564 \emph{``The Design of Modern Fully Dynamic Data  %
Structures (MoDynStruct)''}) and from the Austrian Science Fund (FWF) project Z  %
422-N, and project \emph{``Fast Algorithms for a Reactive Network Layer          %
(ReactNet)''}, P~33775-N, with additional funding from the \textit{netidee       %
SCIENCE Stiftung}, 2020--2024.                                                   %
}                                                                                %

\date{}

\begin{document} %

\maketitle

\begin{abstract}
Representing a polygon using a set of simple shapes has numerous applications in different use-case scenarios.
We consider the problem of covering the interior of a rectilinear polygon with holes by a set of area-weighted, axis-aligned rectangles
such that the total weight of the rectangles in the cover is minimized.
Already the unit-weight case is known to be \NP{}-hard and the general problem has, to the best of our knowledge, not been studied experimentally before.

We show a new basic property of optimal solutions of the weighted problem.
This allows us to speed up existing algorithms for the unit-weight case, obtain an improved ILP formulation for both the weighted and unweighted problem, and develop several approximation algorithms and heuristics for the weighted case.

All our algorithms are evaluated in a large experimental study on \num{186 837} polygons
combined with six cost functions,
which provides evidence that our algorithms are both fast and yield close-to-optimal solutions in practice.

\end{abstract}

\section{Introduction}%

Representing a polygon's interior using simpler shapes is a relevant problem in many fields, including integrated circuit design~\cite{Hegedus_1982}, image compression~\cite{Koch_Marenco_2022} and image construction~\cite{Construction}.
When using rectangles, the main goal is usually to minimize the number of rectangles used to represent the polygon. 
However, in some applications, there may be a cost associated with both the number of shapes used as well as their area.
If shapes may overlap, this can lead to a solution with more shapes costing less than a solution with fewer shapes.

One such example are 2D video games in which sets of individual ``tiles'' may be more compactly represented as ``objects''. 
Each object requires a certain amount of time to initialize itself, after which each of its tiles requires a certain amount of time to be rendered. 
Depending on how long both actions take, it may be faster to have more objects, consisting of fewer total tiles, or fewer objects, consisting of more total tiles. 
Note that for the purpose of this comparison, if two or more tiles overlap, all of them are still rendered and thus contribute to the total runtime.

This motivates our study of the following \textsf{Weighted Rectangle Cover} problem (\WRC):
Given a rectilinear polygon, which may contain holes, together with a cost
function, find a set of rectangles which covers the polygon's interior with minimum total cost.
The problem generalizes the unit-weight case, which is known to be $\NP$-hard~\cite{Masek_1978,DBLP:journals/jal/CulbersonR94}.
Although results on approximability and experimental studies for the
unit-weight case and the closely related rectilinear picture compression problem
exist~\cite{Kumar_Ramesh_1999,DBLP:journals/algorithmica/BermanD97,Heinrich-Litan_Lubbecke_2007,Koch_Marenco_2022},
little is known about the weighted case. 

In this work, we initiate the study of the weighted rectangle cover problem.
Motivated by the applications above, we focus on cost functions
that can be expressed as a linear function of the rectangle's area (`area
cost') plus a constant (`creation cost'), but many of our results can be transferred %
to other cost functions.
\smallskip

\noindent\textbf{Contributions} %
\begin{itemize}
\item
We introduce the concept of base rectangles as an alternative approach to (Hanan) grid rectangles and show that they can reduce the complexity of the discretized problem significantly.
\item
We prove that there always exists an optimal solution that is built from base rectangles.
\item
We give the first polynomially-sized ILP formulation for the \WRC{} problem, which also improves the
ILP formulations for the unit-weight case, due to base rectangles.
\item
We develop and analyze several algorithms to solve the \WRC{} problem quickly that yield, in practice, a high solution quality.
In particular, we introduce four weight-aware postprocessors, which can also be used
to adapt algorithms for the unweighted case.
\item
We report on a large experimental evaluation of ten algorithms on \num{186 837} non-trivial polygons in combination with six
cost functions, which shows that our new algorithms are not only faster than the greedy weighted set cover algorithm, but also produce \emph{close-to-optimal} solutions in practice.
\end{itemize}

\smallskip
\noindent\textbf{Paper Outline.}
\autoref{sect:related} and \autoref{sect:prelim} start with related work and preliminaries.
In \autoref{sect:baserects}, we introduce Base Rectangles. %
We describe our algorithms and postprocessors in \autoref{sect:algorithms}, their experimental
evaluation in \autoref{sect:experiments}, and conclude in \autoref{sect:conclusion}.
\begin{figure}[tb]
    \centering
        \begin{tikzpicture}[thin,fill opacity=0.4,draw opacity=0.6,scale=.25] %
            \draw[fill={rgb,255:red,52; green,208; blue,218}] (1.0, 3.0) rectangle (3.0, 5.0);
            \draw[fill={rgb,255:red,254; green,54; blue,107}] (1.0, 6.0) rectangle (15.0, 10.0);
            \draw[fill={rgb,255:red,124; green,175; blue,27}] (1.0, 11.0) rectangle (3.0, 13.0);
            \draw[fill={rgb,255:red,201; green,198; blue,104}] (2.0, 4.0) rectangle (14.0, 12.0);
            \draw[fill={rgb,255:red,127; green,4; blue,109}] (4.0, 3.0) rectangle (12.0, 13.0);
            \draw[fill={rgb,255:red,20; green,81; blue,72}] (5.0, 2.0) rectangle (11.0, 14.0);
            \draw[fill={rgb,255:red,171; green,201; blue,72}] (6.0, 1.0) rectangle (10.0, 15.0);
            \draw[fill={rgb,255:red,140; green,213; blue,145}] (13.0, 3.0) rectangle (15.0, 5.0);
            \draw[fill={rgb,255:red,193; green,15; blue,62}] (13.0, 11.0) rectangle (15.0, 13.0);            
        \end{tikzpicture}
        \quad
        \begin{tikzpicture}[thin,fill opacity=0.4,draw opacity=0.6,scale=.25] %
            \draw[fill={rgb,255:red,90; green,103; blue,128}] (1.0, 3.0) rectangle (3.0, 5.0);
            \draw[fill={rgb,255:red,239; green,11; blue,144}] (1.0, 6.0) rectangle (2.0, 10.0);
            \draw[fill={rgb,255:red,254; green,70; blue,199}] (1.0, 11.0) rectangle (3.0, 13.0);
            \draw[fill={rgb,255:red,10; green,120; blue,108}] (2.0, 4.0) rectangle (4.0, 12.0);
            \draw[fill={rgb,255:red,228; green,55; blue,108}] (4.0, 3.0) rectangle (5.0, 13.0);
            \draw[fill={rgb,255:red,164; green,107; blue,109}] (5.0, 2.0) rectangle (11.0, 14.0);
            \draw[fill={rgb,255:red,136; green,49; blue,48}] (6.0, 1.0) rectangle (10.0, 2.0);
            \draw[fill={rgb,255:red,194; green,232; blue,215}] (6.0, 14.0) rectangle (10.0, 15.0);
            \draw[fill={rgb,255:red,212; green,192; blue,89}] (11.0, 3.0) rectangle (12.0, 13.0);
            \draw[fill={rgb,255:red,184; green,112; blue,126}] (12.0, 4.0) rectangle (14.0, 12.0);
            \draw[fill={rgb,255:red,151; green,35; blue,97}] (13.0, 3.0) rectangle (15.0, 5.0);
            \draw[fill={rgb,255:red,136; green,30; blue,226}] (13.0, 11.0) rectangle (15.0, 13.0);
            \draw[fill={rgb,255:red,249; green,79; blue,107}] (14.0, 6.0) rectangle (15.0, 10.0);            
        \end{tikzpicture}
        \quad
        \begin{tikzpicture}[thin,fill opacity=0.4,draw opacity=0.6,scale=.25] %
            \draw[fill={rgb,255:red,89; green,97; blue,32}] (1.0, 3.0) rectangle (3.0, 4.0);
            \draw[fill={rgb,255:red,108; green,51; blue,60}] (1.0, 4.0) rectangle (15.0, 5.0);
            \draw[fill={rgb,255:red,64; green,158; blue,238}] (1.0, 6.0) rectangle (15.0, 10.0);
            \draw[fill={rgb,255:red,77; green,55; blue,255}] (1.0, 11.0) rectangle (3.0, 13.0);
            \draw[fill={rgb,255:red,103; green,114; blue,71}] (2.0, 5.0) rectangle (14.0, 6.0);
            \draw[fill={rgb,255:red,173; green,177; blue,128}] (2.0, 10.0) rectangle (14.0, 11.0);
            \draw[fill={rgb,255:red,189; green,86; blue,83}] (3.0, 11.0) rectangle (13.0, 12.0);
            \draw[fill={rgb,255:red,116; green,31; blue,9}] (4.0, 3.0) rectangle (12.0, 4.0);
            \draw[fill={rgb,255:red,30; green,189; blue,63}] (4.0, 12.0) rectangle (12.0, 13.0);
            \draw[fill={rgb,255:red,116; green,162; blue,15}] (5.0, 2.0) rectangle (11.0, 3.0);
            \draw[fill={rgb,255:red,161; green,47; blue,30}] (5.0, 13.0) rectangle (11.0, 14.0);
            \draw[fill={rgb,255:red,32; green,153; blue,206}] (6.0, 1.0) rectangle (10.0, 2.0);
            \draw[fill={rgb,255:red,245; green,221; blue,226}] (6.0, 14.0) rectangle (10.0, 15.0);
            \draw[fill={rgb,255:red,89; green,163; blue,188}] (13.0, 3.0) rectangle (15.0, 4.0);
            \draw[fill={rgb,255:red,126; green,105; blue,18}] (13.0, 11.0) rectangle (15.0, 13.0);           
        \end{tikzpicture}
\caption{Three optimal covers $C_1$ (left), $C_2$ (center), $C_3$ (right) for the same polygon with different parameters. \\
For $C_1$, $\beta = 0$, $|C_1| = 9$, and $\sum_{R \in C_1}{\Area(R)} = 376$. \\
For $C_2$, $\frac \beta \alpha = \frac 1 3$, $|C_2| = 13$, and $\sum_{R \in C_2}{\Area(R)} = 156$.\\
For $C_3$, $\frac \beta \alpha  = 2$, $|C_3| = 15$, and $\sum_{R \in C_3}{\Area(R)} = 152$.
}%
\label{fig:area-objective}
\end{figure}

\section{Related Work on Rectilinear Polygons}\label{sect:related}%
Unless denoted otherwise, we assume in the following that the polygon under
consideration is rectilinear and all rectangles are axis-aligned.
\paragraph{Weighted Set Cover.}
Given $k$ subsets $\mathcal{S} = (S_1, S_2, \ldots, S_k)$ of a universe $U$,
each of which has an associated weight $w(S_j)$,
this problem asks to pick a subset $\mathcal{S'} \subseteq \mathcal{S}$ whose union equals $U$
and minimizes $\sum_{S \in \mathcal{S'}} w(S)$.
Simple greedy $\bigO(\log{|U|})$-approximation algorithms exist both for the
unweighted~\cite{Lovasz_1975,Johnson_1973} and for the weighted
case~\cite{Chvatal_1979}.
Feige~\cite{Feige_1998} showed that this approximation ratio is tight unless
$\Poly = \NP$.
\emph{Weighted geometric set cover} refers to a family of special
cases, where the subsets $S_j$ are geometric
objects~\cite{Varadarajan_2010}. %
Even \ea{}~\cite{EVEN2005358} give a randomized
$\bigO(\log{\OPT})$-approximation algorithm for the special case of
geometric objects with so-called ``low VC-dimension'' %
with all weights greater than $1$. %
Their algorithm is based on solving the relaxed linear program of the problem
instance and then rounding it by probabilistically finding an $\epsilon$-net
of small size.
As $\OPT$ can be arbitrarily larger than $|U|$~\cite{Agarwal_Ezra_Fox_2022},
this is not necessarily an improvement over the simple greedy algorithm.
Varadarajan~\cite{Varadarajan_2010} gives a similar probabilistic algorithm
which uses a ``quasi-uniform'' sampling of $\epsilon$-nets to round the linear
program solution and has an approximation ratio of $2^{\bigO(\log^*{k})} \log{k}$,
with $k = |\mathcal{S}|$.
This improves over previous results if the \emph{union
complexity}, %
\ie, the complexity of the boundary of the union of all geometric objects,
is near-linear.
Note that the union complexity is in $\bigO(r^2)$ for $r$ rectangles~\cite{DBLP:journals/combinatorics/KellerS18}.

There are also interesting recent results for the unweighted case~\cite{Agarwal_Pan_2014,Bus_Mustafa_Ray_2018},
which are likewise centered around the computation of $\epsilon$-nets and largely theoretical, though an efficient
computation of $\epsilon$-nets for disks has been implemented~\cite{Bus_Mustafa_Ray_2018}.

\paragraph{Rectilinear Polygon Partition.}
A \emph{partition} (also called \emph{decomposition} or \emph{dissection}) of a rectilinear polygon
into a minimum set of non-overlapping rectangles can, in contrast to the
overlapping case, be computed in polynomial time even if the polygon contains
holes.
For a polygon with $n$ vertices,
Ohtsuki~\cite{Ohtsuki_1982} gives an algorithm with running time
$\bigO(n^{2.5})$, which was later improved to $\bigO(n^{1.5} \log{n})$~\cite{Hiroshi_Takao_1986}.
The algorithm was in fact discovered several
times~\cite{data-organization-2d,DBLP:journals/cvgip/FerrariSS84,DBLP:conf/wg/Eppstein09} and is
described in more detail in \autoref{sect:partition}.

A related problem, where instead of minimizing the number of rectangles, the
goal is to minimize the total length of all segments used to create the
partition is studied by Lingas \ea{}~\cite{LPRS1982}.
They show that the problem can be solved in time $\bigO(n^4)$ for arbitrary rectilinear hole-free polygons with $n$ vertices, and in $\bigO(n^3)$ for ``histogram'' polygons.
In the presence of holes, the problem becomes \NP-hard.

\paragraph{Unweighted Rectangle Cover.}
The problem of covering a rectilinear polygon with holes by axis-aligned
rectangles was proven to be $\NP$-hard by Masek~\cite{Masek_1978} in 1978,
and remains $\NP$-hard for hole-free polygons~\cite{DBLP:journals/jal/CulbersonR94}.

Assume that the input polygon has $n$ vertices (including the vertices of holes).
Since the unweighted problem is known to be a special case of the set cover problem with a universe of size
polynomial in $n$, it can be approximated likewise within an $\bigO(\log{n})$ factor.
Furthermore, a polynomial-time $\bigO(\sqrt{\log{n}})$-approximation algorithm
exists, which is due to Kumar and Ramesh~\cite{Kumar_Ramesh_1999}.
In \autoref{sect:greedy-strip}, we will give a modified version of this algorithm
and show that it can be implemented with running time
$\bigO(n^2)$. %
Berman and DasGupta~\cite{DBLP:journals/algorithmica/BermanD97} showed that the
rectangle cover problem does not admit a polynomial-time approximation scheme
unless $\Poly = \NP$. %

Improved results exist for certain special cases:
For polygons that are vertically convex, Franzblau and
Kleitman~\cite{DBLP:journals/iandc/FranzblauK84} give a $\bigO(n^2)$-time
algorithm that computes an optimal solution.
A polygon is vertically (horizontally) convex if the straight-line segment
connecting any two interior points that have the same $x$-coordinate
($y$-coordinate) lies in the interior of the polygon.
In particular, this implies that the polygon is hole-free.
Franzblau~\cite{DBLP:journals/siamdm/Franzblau89} showed that
for a polygon with $h$ holes, a partition has at most $2\cdot\OPT+ h -1$ rectangles,
thus giving a $2$-approximation algorithm for hole-free polygons.
She also gives a $\bigO(n \log n)$-time algorithm that produces a rectangle
cover of size $\bigO(\OPT\cdot \log \OPT)$. 
The problem was also studied in~\cite{Najman14}.

\paragraph{Rectilinear Picture Compression.}
Given a binary matrix, the \emph{rectilinear picture compression} problem asks
to represent it as a minimum set of rectangular submatrices containing only
one-entries, such that every one-entry is contained in at least one of our
submatrices.
This is equivalent to an unweighted and integral version of the rectangle cover
problem, where in addition, the size of the input equals the \emph{area} of the
bounding box of the represented polygon, \ie, $w \cdot h$
if $w$ and $h$ denote the width and height of the input polygon's bounding box.
Note that $w$ and $h$ are \emph{not} bounded by the number of vertices of the
polygon in general.

Heinrich-Litan and L\"{u}bbecke~\cite{Heinrich-Litan_Lubbecke_2007} %
give primal and dual integer linear programming formulations for this problem,
further ones are discussed by Koch and Marenco~\cite{Koch_Marenco_2022}.
An algorithm that produces good results in practice is
presented by the same authors~\cite{Heinrich-Litan_Lubbecke_2007}. 
It is based on the greedy set cover algorithm (cf.~\autoref{sect:greedy-set}), but
extends it by picking certain rectangles which are guaranteed to be part of
some optimal cover
before invoking the greedy set cover algorithm to cover the
remaining parts of the polygon, pruning any redundant rectangles in a
postprocessing step.
This algorithm is based on an earlier one~\cite{Hannenhalli_Hubell_Lipshutz_Pavel_2002}, 
which has worst-case runtime complexity $\bigO(\max\{w,h\}^5)$.

Fast approaches for large instances with potentially improved quality are
discussed by Koch and Marenco~\cite{Koch_Marenco_2022}.
Their approaches work by finding an initial cover for the polygon by looking
for large rectangles, grouping the pixels of the polygon into so-called atomic
rectangles using this initial cover and then choosing a subset of the initial
cover to cover the obtained atomic rectangles.

\emph{To the best of our knowledge, no specific algorithms or results for the
weighted rectangle cover problem exist.}

\section{Preliminaries}\label{sect:prelim}%
We consider a non-self-intersecting, rectilinear polygon $\Polygon = (\Edges, \Holes)$, which is specified by a
cyclic, alternating list of horizontal and vertical \emph{edges} $\Edges$ 
along with a set of \emph{holes} $\Holes$, where each hole $H \in \Holes$ is
itself a hole-free, non-self-intersecting, rectilinear polygon that is strictly contained in $\Polygon$.
A polygon is said to be \emph{hole-free} if $\Holes = \emptyset$.
An \emph{edge}
$e = \Edge{u}{v}$
is either a horizontal or vertical segment
connecting the two \emph{vertices} $u$ and $v$, where $u, v \in \Rplus^2$.
A vertex $v$ is \emph{integral} if $v \in \N^2$.
For every pair of subsequent edges
$e_1 = \Edge{u_1}{v_1}$ and $e_2 = \Edge{u_2}{v_2}$ from $E$, if
$v_1 = u_2$ then $e_1$ and $e_2$ are called \emph{adjacent}.
Thus, $E$ is fully determined by either the subsequence of horizontal %
or the subsequence of vertical edges.
A pair of non-adjacent edges $e_1, e_2$ is said to be \emph{intersecting} if the intersection
of their segments is nonempty.
A polygon is \emph{self-intersecting} if it has a pair of intersecting edges.

In the following, we assume that all polygons are rectilinear,
non-self-intersecting, and that two holes of the same polygon may only
intersect in a single common vertex%
\footnote{Note that a self-intersecting
polygon can be represented by a set of non-self-intersecting polygons and that
two intersecting holes can be represented by a single, larger hole.}.

The \emph{interior} $\Interior(\Polygon)$ of a polygon $\Polygon = (\Edges, \Holes)$ is
the bounded region enclosed by $\Edges$ minus the interiors of all holes in $\Holes$.
A polygon $C$ is said to be \emph{(strictly) contained} in a polygon $P$
if $\Interior(C)$ is a
(strict) subset of $\Interior(P)$.
The \emph{area} $\Area(\Polygon)$ is the area of $\Interior(\Polygon)$ in $\R^2$.
A point $x \in \Interior(\Polygon)$ is an \emph{inner point} if $x$ is not
part of the boundary of neither $\Polygon$ nor any of its holes.
A vertex of a polygon is called \emph{convex} if the interior angle between its two edges is
$\frac{\pi}{2}$ and \emph{concave} otherwise.
Note that the definition of convex and concave depends on the interior, \ie,  a convex vertex of a polygon $H$ may be
concave w.r.t.\ another polygon $P$ if $H$ is added as a hole to $P$, see also \autoref{fig:concave-convex}.
\begin{figure}[tb]
    \centering
    \includegraphics[scale=.9]{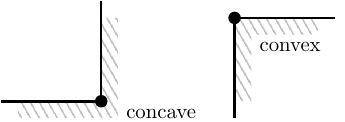} %
    \caption{Concave and convex vertex w.r.t.\ the interior of the polygon.}\label{fig:concave-convex}
\end{figure}

For $\Polygon = (\Edges, \Holes)$,
we use $\AllEdges(\Polygon)$ and $\AllVertices(\Polygon)$ to refer to the union
of all edges and vertices in $\Edges$ and $\Holes$, respectively,
\ie, $\AllEdges(\Polygon) = \Edges \cup \bigcup_{H \in \Holes} \AllEdges(H)$,
and $\AllVertices(\Polygon)$ is the set of all vertices in $\AllEdges(\Polygon)$.
We let $n(\Polygon)$ denote the \emph{size} of $\Polygon$, where $n(\Polygon) := \frac{1}{2}|\AllEdges(\Polygon)|$, \ie,
the number of horizontal (or vertical) edges in $\AllEdges(\Polygon)$.
For brevity, we just use $n$ if no ambiguity arises.
Note that $|\AllVertices(\Polygon)| \leq 2n$ due to vertex-intersecting holes.

A \emph{rectangle} $\Rect$ is a hole-free polygon with $n(\Rect) = 2$.
For simplicity, we represent $R$ as a pair of vertices $(x_1, x_2)$,
where $x_1 =: \TopLeft(\Rect)$ and $x_2 =: \BottomRight(\Rect)$ denote the top-left and
bottom-right vertex of $\Rect$, respectively.
The \emph{bounding box} $\BBox(\Polygon)$ %
of a polygon $P$ is the rectangle $R$ such that $\Interior(P) \subseteq 
\Interior(R)$ 
and $\Area(R)$ is minimal.
A \emph{pixel} is a unit-area, square rectangle with only integral vertices. %
A rectangle $M$ is said to be a \emph{maximal} rectangle of $\Polygon$ 
if $\Interior(M) \subseteq \Interior(\Polygon)$ and extending $M$ in any
direction would violate the former constraint.
We denote the set of all maximal rectangles of $\Polygon$ by $\MaxRects(\Polygon)$.

Given a polygon $\Polygon$, a set of rectangles $\Rectangles$ is a \emph{rectangle cover}
(or \emph{cover} for short)
if (i) for each $\Rect \in \Rectangles$, $\Interior(R) \subseteq \Interior(\Polygon)$
and (ii) $\Interior(\Polygon) = \bigcup_{\Rect \in \Rectangles} \Interior(R)$.
A cover is a \emph{partition} if the pair-wise intersections of
the rectangles contain no inner points.
For a set of rectangles $\Rectangles$ and rectangle $\Rect$ that is not necessarily
part of $\Rectangles$, we use the shorthand notation
$\Rectangles_{\supseteq \Rect} := \{ \Rect' \in \Rectangles \colon \Interior(\Rect') \supseteq \Interior(\Rect) \}$.

We study the \textsc{Weighted Rectangle Cover} problem (\WRC{}) as follows:
Given a (non-self-intersecting, rectilinear) polygon $\Polygon = (\Edges, \Holes)$,
as well as $\alpha, \beta \in \Rplus$,
find a rectangle cover $\Cover$ of $\Polygon$ that minimizes
$\sum_{\Rect \in \Cover} \Cost_{\alpha,\beta}(\Rect)$,
where the \emph{cost} of a rectangle $\Rect$ is $\Cost_{\alpha,\beta}(\Rect) := \alpha + \beta\cdot\Area(\Rect)$.

Note that ratio $\frac \beta \alpha$ allows for a smooth transition between the objective of the $\NP$-hard unit-weight Rectangle-Cover problem ($\beta =0$) and the Rectilinear Polygon Partition problem ($\beta \to \infty$), which can be solved in polynomial time, see \autoref{fig:area-objective} for examples.

\section{Base Rectangles and Rectangle Powersets}\label{sect:baserects}%
In this section, we introduce a coarse discretization into \emph{base rectangles} and prove that there is always an optimal \WRC{} solution that can be derived from them.
We show that such discretizations can be computed efficiently, which incidentally allows us to improve upon pixel-based and grid-based approaches for the unit-weight case.
Further, their use leads to speedups for various algorithms for the weighted and unweighted problem on real-world instances, which we empirically provide evidence for in Section~\ref{sect:experiments}. %

Before we introduce our new concept of base rectangles, we give the definition
of grid rectangles, also known as \emph{Hanan grid}.

\begin{Definition}[Grid Rectangles~\cite{Kumar_Ramesh_1999}]
Extend each edge $e \in \AllEdges(\Polygon)$ of a polygon $\Polygon$ to infinity on both ends and consider all intersections of all these lines.
This partitions $\Interior(\Polygon)$ into a set of grid rectangles $\GridRects(\Polygon)$.
\end{Definition}

\begin{Definition}[Base Rectangles]
From each concave vertex $v \in \AllVertices(\Polygon)$ of a polygon $\Polygon$, extend one horizontal and one vertical ray from $v$ through $\Interior(P)$ until it meets with an edge in $\AllEdges(\Polygon)$.
This partitions $\Interior(\Polygon)$ into a set of base rectangles $\BaseRects(\Polygon)$.
\end{Definition}

Note that grid rectangles are a refinement of base rectangles.
\autoref{fig:base-rects} shows two examples with subdivisions of a polygon into base rectangles,
which in one case coincides with a subdivision into grid rectangles, but makes a large
difference in the other.
A polygon with only integral vertices can always be partitioned
into a set of pixels and its set of base rectangles (grid rectangles) is a coarsening
of this pixelwise partition.  %
\begin{figure}[]
\centering
    \includegraphics[width=.8\columnwidth]{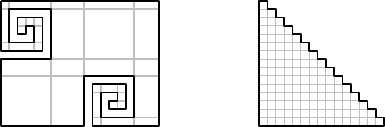} %

\caption{Two polygons (without holes) and their partition into base rectangles. %
The left shows that the number of base rectangles can be significantly smaller than the number of grid rectangles,
whereas both coincide on the right.}%
\label{fig:base-rects}
\end{figure}

\begin{lemma}\label{lem:num_baserects}
Given a polygon $\Polygon$ with $2n$ edges, the number of base rectangles $|\BaseRects(\Polygon)|$ is $\bigO(n^2)$ in the worst-case.
\end{lemma}
\begin{proof}
As $|\AllVertices(\Polygon)| \leq 2n$,
we draw at most $2n$ horizontal and $2n$ vertical lines.
Each base rectangle is uniquely identified by its top left
vertex, which is formed by the intersection of a horizontal and
a vertical line.
Thus, there are at most $2n \cdot 2n = 4n^2$ base rectangles.
\end{proof}

\begin{lemma}\label{lem:compute_baserects}
Given a polygon $\Polygon$ with $2n$ edges, the base rectangles $\BaseRects(\Polygon)$ can be computed in $\bigO(n \log n + |\BaseRects(\Polygon)|)$ time. 
\end{lemma}
\begin{proof}
We first compute the arrangement, using the well-known sweep line algorithm, in $\bigO(n \log n)$ time.
Within the same time bound, we can build a vertical ray-shooting data structure for the arrangement, and one for horizontal ray-shooting.
Based on those, one can determine the vertical, and horizontal, extension ray that emits from a concave vertex in $\bigO(\log n)$ time, taking $\bigO(n \log n)$ total time.
Finally, we compute the $\bigO(|\BaseRects(\Polygon)|)$-sized arrangement, which includes the intersection points of the rays' line segments, in $\bigO(n \log n + |\BaseRects(\Polygon)|)$ time, using~\cite{ChazelleE92} for example.
\end{proof}

\begin{Definition}[Rectangle Powerset]
For a set of rectangles $\Rectangles$, the powerset $\RectPow(\Rectangles)$ is the set of all rectangles that can be covered by a subset from $\Rectangles$.
\end{Definition}
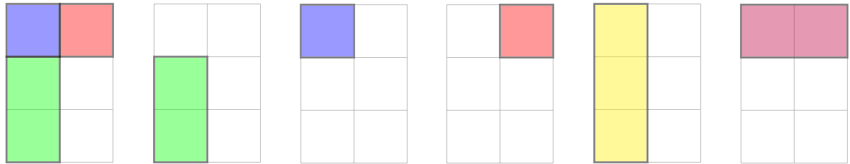
\begin{figure}[b]
\centering
    \begin{tikzpicture}[thick,fill opacity=.4,draw opacity=.4,scale=.7] %
        \draw[step=1cm,gray,very thin] (1, 0) grid (3, 3);
        \draw[fill=green] (1, 0) rectangle (2, 2);
        \draw[fill=blue] (1, 2) rectangle (2, 3);
        \draw[fill=red] (2, 2) rectangle (3, 3);
    \end{tikzpicture}
    \quad
    \begin{tikzpicture}[thick,fill opacity=.4,draw opacity=.4,scale=.7] %
        \draw[step=1cm,gray,very thin] (1, 0) grid (3, 3);
        \draw[fill=green] (1, 0) rectangle (2, 2);
    \end{tikzpicture}
    \quad
    \begin{tikzpicture}[thick,fill opacity=.4,draw opacity=.4,scale=.7] %
        \draw[step=1cm,gray,very thin] (1, 0) grid (3, 3);
        \draw[fill=blue] (1, 2) rectangle (2, 3);
    \end{tikzpicture}
    \quad
    \begin{tikzpicture}[thick,fill opacity=.4,draw opacity=.4,scale=.7] %
        \draw[step=1cm,gray,very thin] (1, 0) grid (3, 3);
        \draw[fill=red] (2, 2) rectangle (3, 3);
    \end{tikzpicture}
    \quad
    \begin{tikzpicture}[thick,fill opacity=.4,draw opacity=.4,scale=.7] %
        \draw[step=1cm,gray,very thin] (1, 0) grid (3, 3);
        \draw[fill=yellow] (1, 0) rectangle (2, 3);
    \end{tikzpicture}
    \quad
    \begin{tikzpicture}[thick,fill opacity=.4,draw opacity=.4,scale=.7] %
        \draw[step=1cm,gray,very thin] (1, 0) grid (3, 3);
        \draw[fill=purple] (1, 2) rectangle (3, 3);
    \end{tikzpicture}
\caption{A set of three interior-disjoint rectangles (far left) and all five rectangles in its rectangle powerset.}%
\label{fig:rect-powerset}
\end{figure}

\autoref{fig:rect-powerset} shows an example of a rectangle powerset that is generated by three rectangles.
Recall that for the unit-weight rectangle cover problem, it is well known that there is always an optimal solution that consists of maximal rectangles~\cite{DBLP:journals/siamdm/Franzblau89}, and that the set of maximal rectangles $\MaxRects(\Polygon)$ has size $\bigO(n^2)$ in the worst-case.
Observe that the powerset of the base rectangles $\RectPow(\BaseRect(\Polygon))$ always contains the maximal rectangles $\MaxRects(\Polygon)$.
Though our rectangle powersets can be much larger, their size remains polynomial.
\begin{lemma}\label{lem:rectpow_size}
Given a set of $r$ rectangles $\Rectangles$, $|\RectPow(\Rectangles)| \in \bigO(r^2)$.
\end{lemma}
\begin{proof}
Let $V$ be the set of all top left and bottom right vertices of rectangles in $\Rectangles$.
Then, $|V| \leq 4r$.
Each rectangle that can be covered by a subset of $\Rectangles$ must have its top left and bottom right vertex in $V$, so at most $\binom{4r}{2} \in \bigO(r^2)$ different rectangles can be covered by $\Rectangles$.
\end{proof}

\begin{corollary}\label{cor:rectpow_size}
Given a polygon $\Polygon$ with $2n$ edges, the rectangle powerset of its base rectangles contains %
at most
$|\RectPow(\BaseRects(\Polygon))| \in \bigO({|\RectPow(\BaseRects(\Polygon))|}^2) \subseteq \bigO(n^4)$ rectangles.
\end{corollary}

We now show that it suffices to only consider base rectangles as candidates for
an optimal cover in the \emph{weighted} problem setting, which also constitutes the main result of this section.
The following result is well-known for the unit-weight case, but its correctness
depends on the cost function for the general case.
We show that it holds for our cost functions and use it as building block
for our main result in Theorem~\autoref{thm:base-aligned}.
(Similar proofs~\cite{Zachariasen01, SeyboldF14} are known for other problems.)

\begin{lemma}[Grid-aligned]\label{lem:grid-aligned}
    For every polygon $\Polygon$ and reals $\alpha, \beta \geq 0$, there
    exists an optimal weighted rectangle cover $\Cover$ such that $\Cover \subseteq
\RectPow(\GridRects(\Polygon))$, where $\GridRects(\Polygon)$ denotes the set of grid cells defined by the vertices of $\Polygon$.
\end{lemma}
\begin{proof}
Let $\Pi_x( Q )$ be the set of (distinct) $x$-coordinates of a set of points $Q$, and $\Pi_y(Q)$ its set of $y$-coordinates.
Recall that the vertices of the Hanan-grid $\GridRects(\Polygon)$ are given by the Cartesian product $\Pi_x( \AllVertices(\Polygon))\times \Pi_y( \AllVertices(\Polygon) )$.

Using a sweep argument with a vertical line, we will show that, for any cover $\Cover''$, there exists a cover $\Cover'$ with $\Pi_x(\AllVertices(\Cover')) \subseteq \Pi_x(\AllVertices(\Polygon))$ and cost \emph{no larger} than~$\Cover''$.
Using the symmetric sweep argument with a horizontal line on $\Cover'$ shows the lemma.
Note that this argument applies in particular to an optimal cover $\Cover''$.

The sweep argument proceeds with ascending {$x$-co}ordinates and maintains a cover $\Cover'$ that remains feasible at all times, i.e. at all $x$-events.
Initially, we set $\Cover':=\Cover''$ and maintain the sweep-invariant that all vertices of $\Cover'$ that are to the left of the sweep line are contained in $\Pi_x(\AllVertices(\Polygon))$.
It remains to specify how to update $\Cover'$ at the $x$-events, i.e. for all $x$-values in $\Pi_x(\AllVertices(\Cover'')) \setminus \Pi_x(\AllVertices(\Polygon))$.

For a given $x$-event, consider all rectangles in the current cover $\Cover'$ that have a vertical edge that is contained in the sweep line.
Let $h^-$ denote the total (summed) height of all incident rectangles that are to the left of the sweep line,
and $h^+$        the total          height of all incident rectangles that are to its right.
Now, jointly moving all vertical sides of those rectangles by a small $\delta \geq 0$, i.e. to the right, or by a small $\delta<0$, i.e. to the left, changes the cost of the rectangle set $\Cover'$ exactly by 
\begin{equation}\label{eq:hanan-sweep}
\beta ~\cdot~ \delta (h^- - h^+)~.
\end{equation}
Thus, there is one direction that does not increase cost.
Note that this is true for all $\delta \in [x^- - x, x^+-x]$, where 
$|x^--x|$ is the distance of the sweep line to its predecessor value in $\Pi_x(\AllVertices(\Polygon))$
and $x^+-x$ is the distance of the line to its successor value in $\Pi_x(\AllVertices(\Cover'')) \cup \Pi_x(\AllVertices(\Polygon))$.
Clearly, deleting all rectangles $R$ from $C'$ that become degenerate by moving, i.e. $\Area(R)=0$, does not increase cost, since $\alpha\geq 0$.
It remains to show that the set of rectangles $\Cover'$ remains a valid cover of the polygon interior for those values of $\delta$.

To see that the cover remains feasible when moving the boundaries to the right ($\delta\geq 0$), consider an arbitrary incident rectangle $R$ that lies to the right of $x$. ($R$ is shortened by the move.)
If its left boundary $\mathit{left}(R)$ is covered by rectangles that are to the left of the sweep line and incident to it, then there is nothing to show.
If there is a point $q \in \mathit{left}(R)$ that is not covered by a rectangle that is to the left of the sweep line and incident to it, then we have that $q$ is in the interior of the polygon.
This is due to the $x$-events being from $\Pi_x(\AllVertices(\Cover'')) \setminus \Pi_x(\AllVertices(\Polygon))$ only, i.e., those $x$-coordinates of $\Cover''$ that are different from those of the polygon.
Let $q'\in \Interior(\Polygon)$ have the same $y$-coordinate as $q$, but an infinitesimally smaller $x$-coordinate.
Since the initial set $\Cover''$ is a cover of the polygon's interior $\Interior(\Polygon)$, there is a rectangle $R' \in \Cover'$ that is \emph{not incident} to the sweep line and contains $q' \in \Interior(R')$. 
Thus, $q \in \Interior(R')$ and remains covered.
The argument for moving the vertical boundaries to the left ($\delta <0$) is analog.

Using the symmetric sweep argument with a horizontal line on $\Cover'$ shows the lemma.
\end{proof}

\begin{theorem}[Base-Aligned]\label{thm:base-aligned}
For every polygon $\Polygon$ and reals $\alpha, \beta \geq 0$, there
exists an optimal weighted rectangle cover $\Cover$ such that $\Cover \subseteq
\RectPow(\BaseRects(\Polygon))$.
\end{theorem}

\begin{proof}
From \autoref{lem:grid-aligned}, we have that there is an optimal solution $\Cover''$ whose rectangles consist of vertices from the Hanan-grid $\GridRects(\Polygon)$.

Using a sweep argument with a vertical line, we will show that, for any such cover $\Cover''$, there exists a cover $\Cover'$ with, cost \emph{no larger} than~$\Cover''$, whose rectangles have vertical sides that are either contained in the boundary of $\Polygon$ or in the ray of a concave vertex.
Using the symmetric sweep argument with a horizontal line shows the theorem.

As in the previous proof, our invariant will be that all vertical boundaries that are to the left of the sweep line have this property.
The $x$-events of the sweep line are however from the grid $\Pi_x(\AllVertices(\Polygon))$.

For a given $x$-event, consider the vertical chords through the interior $\Interior(\Polygon)$ with the given $x$-coordinate (i.e. the intersection of the sweepline and $\Interior(\Polygon)$).
Note that chords are either separated by a section of the boundary or by the exterior of polygon $\Polygon$.
For a given chord, consider all rectangles of $\Cover'$ that have a vertical boundary contained in the chord.
If the chord contains a concave vertex of the boundary of $\Polygon$, the incident boundaries of the rectangles are already as desired.
For the case that the chord does not contain a concave vertex, we observe that this is only possible if the chord contains no vertex of the polygon $\AllVertices(\Polygon)$ at its top, or at its bottom.
Thus, all rectangles whose vertical sides are incident with the chord can be moved to the predecessor or successor value in $\Pi_x(\AllVertices(\Polygon))$ without increasing the cost.
The change in cost of $\Cover'$ for moving, to the left or right, is again given by \autoref{eq:hanan-sweep}.
Thus, moving in one of the directions has non-increasing cost.
Moreover, since the chord at the predecessor $x$-coordinate has the identical $y$-range,
all vertical boundaries that are incident in the chord may be moved to the left until there is a concave vertex on the chord.
We apply this argument separately to each of the vertical chords that are contained in the sweep line. %
After all $x$-events are resolved, we have obtained a cover $\Cover'$ with the stated property for vertical boundaries.
\end{proof}

\begin{lemma}\label{lem:min_baserect}
Let $A_{\min}$ be the minimum area of any base rectangle of a given polygon $\Polygon$.
Then, a partition is an optimal cover of $\Polygon$ with cost function $\Cost_{\alpha,\beta}$
for $\alpha, \beta \geq 0$ if $\beta \cdot A_{\min} \geq \alpha$.
\end{lemma}
\begin{proof}
By Theorem~\autoref{thm:base-aligned}, $\Polygon$ has an optimal cover $\Cover$ for cost function $\Cost_{\alpha,\beta}$
such that $\Cover \subseteq \RectPow(\BaseRects(\Polygon))$.
As $\Cover$ is optimal, no edge of a rectangle $\Rect \in \Cover$ is fully contained in another rectangle
$\Rect'$, otherwise, $\Rect$ could be removed from $\Cover$ or its area reduced to avoid an overlap
with $\Rect'$ entirely (see also \autoref{sect:postproc}).

If $\Cover$ is not a partition, there are at least two rectangles in $\Cover$ with an overlap of at least $A_{\min}$.
Removing that overlap by subdividing one of the two rectangles creates one new
rectangle and thus changes the cost of the cover by $\alpha - \beta \cdot A_{\min} \leq 0$.
Repeating this step until no overlap remains results in a partition $\Cover'$ of cost no larger than that of $\Cover$.
\end{proof}

To simplify the description of our algorithms, we introduce the \emph{base
rectangle graph} $\BaseRectGraph(\Polygon)$ of a polygon $\Polygon$, where
each node $\BaseRect$ is a base rectangle and has up
to four labelled neighbors \emph{left}~$\ell(\BaseRect)$, \emph{top}~$t(\BaseRect)$,
\emph{right}~$r(\BaseRect)$, \emph{bottom}~$b(\BaseRect)$, which are the base rectangles that
$\BaseRect$ shares its left, top, right, or bottom edge with, respectively.
We say that a path in $\BaseRectGraph(\Polygon)$ is \emph{unidirectional} if
it starts at some node and continues only via either left, or top, or right, or bottom neighbors.

\begin{lemma}\label{lem:compute_baserectgraph}
Given a polygon $\Polygon$ with $2n$ edges, the base rectangle graph $\BaseRectGraph(\Polygon)$ can be constructed in 
$\bigO(n \log n + |\BaseRects(\Polygon)|)$ 
time and has $\bigO(|\BaseRects(\Polygon)|)$ nodes.
There are at most $\bigO(n)$ nodes that have less than four neighbors and each unidirectional path has length $\bigO(n)$.
\end{lemma}
\begin{proof}
The construction and the number of nodes is directly implied by
\autoref{lem:compute_baserects}.
Each node that does not have four neighbors must be a base rectangle where at least one
of its sides is a segment of an edge in $\AllEdges(\Polygon)$.
By construction (cf.~proof of \autoref{lem:num_baserects}), each vertex can cause at most
one subdivision of a horizontal and one subdivision of a vertical edge.
As there are $\leq 2n$ vertices, the total number of horizontal and vertical
edge segments is $\leq 2n$ each.
The same argument implies that each unidirectional path has length $\leq 2n$.
\end{proof}

\begin{lemma}\label{lem:enum_rectpow}
Given the base rectangle graph $\BaseRectGraph(\Polygon)$, the powerset $\RectPow(\BaseRects(\Polygon))$ can be
reported in $\bigO(|\RectPow(\BaseRects(\Polygon))|)$ time.
\end{lemma}
\begin{proof}
We first compute the base rectangle graph $\BaseRectGraph(\Polygon)$ %
using \autoref{lem:compute_baserectgraph}.

For each node $\BaseRect$ of $\BaseRectGraph$, we enumerate all rectangles $\Rect$ that have
$\TopLeft(\Rect) = \TopLeft(\BaseRect)$ by starting with $b(\BaseRect)$ and first
moving to the respective bottom neighbor in each step for as long as possible.
In each step, report the rectangle with top left vertex $\TopLeft(\BaseRect)$
and bottom right vertex equal to the bottom right vertex of the current base
rectangle.
Let $\min_y$ be the $y$-coordinate of the bottom right vertex of the last base rectangle.
Repeat the same process starting with $r(\BaseRect)$, until either a base rectangle whose bottom right
vertex has $y$-coordinate $\min_y$ is reached or the current base rectangle has no bottom neighbor.
In the second case, update $\min_y$ to the $y$-coordinate of this last base rectangle.
Repeat the process with $r(r(\BaseRect))$, $r(r(r(\BaseRect)))$ and so on until a node without a right neighbor is reached.
\end{proof}

\section{Algorithms}\label{sect:algorithms}%
In this section, we present and analyze three heuristic base algorithms and
six postprocessing routines to solve the \WRC{} problem quickly in
practice.
In addition, we give an exact approach in form of a binary linear
program, which we also use to evaluate the solution qualities in our
experimental evaluation.
\subsection{Binary LP (BIP)}\label{sect:bip} %
Our binary linear program is based on the formulation used for rectilinear
picture compression~\cite{Heinrich-Litan_Lubbecke_2007}, which essentially
solves the unweighted rectangle cover problem (cf.~\autoref{sect:related}).
Recall that a crucial difference between the unweighted and the weighted
problem is that we cannot restrict ourselves to consider only maximal rectangles
as candidates for the cover.
Our formulation is based on base rectangles and reads as follows:
\begin{alignat}{2}
\min{}\sum_{\Rect ~ \in~ \RectPow(\BaseRects(\Polygon))} & x_{\Rect} \cdot \Cost_{\alpha,\beta}(\Rect) \label{eq:objective} \\
\textrm{s.\,t.} \sum_{\Rect ~ \in~  \RectPow(\BaseRects(\Polygon))_{\supseteq \BaseRect}} & x_{\Rect} \geq  1 && \forall \BaseRect \in \BaseRects(\Polygon) \label{eq:constraints} \\
    & x_{\Rect} \in \{0,1\} \qquad && \forall \Rect \in \RectPow(\BaseRects(\Polygon))
\end{alignat}
We have one binary variable $x_{\Rect}$ per element of
$\RectPow(\BaseRects(\Polygon))$ and one constraint per base rectangle, where
the constraints in~\eqref{eq:constraints} ensure that every base rectangle is
covered by at least one selected rectangle.
By \autoref{cor:rectpow_size} and \autoref{lem:num_baserects}, we thus have $\bigO(n^4)$ variables
and $\bigO(n^2)$ constraints.

Let $\MaxRects(\Polygon)$ be the set of maximal rectangles of $\Polygon$.
To obtain a formulation for the unweighted rectangle cover problem, it suffices
to only have a binary variable for each rectangle in $\mathcal{M}(\Polygon)$.
In addition, we set a unit cost function in~\eqref{eq:objective}
and replace the constraints~\eqref{eq:constraints} by 
\begin{alignat}{2}
\tag{\ref*{eq:constraints}$'$}
\sum_{M \in {\MaxRects(\Polygon)}_{\supseteq \BaseRect}} & x_{\Rect} \geq  1 &&\qquad \forall \BaseRect \in \BaseRects(\Polygon)
\end{alignat}
Besides giving a formulation for polygons where not all vertices are integral,
the number of constraints also remains $\bigO(n^2)$ and thus only depends polynomially
on the size of the polygon (\ie, the number of edges), as opposed to its area
in the original formulation~\cite{Heinrich-Litan_Lubbecke_2007}.

Note that using $\MaxRects(\Polygon)$ instead of $\RectPow(\BaseRects(\Polygon))$ for
the weighted problem leads to incorrect solutions (e.g.~\autoref{fig:area-objective}).

\subsection{Greedy Weighted Set Cover}\label{sect:greedy-set} %
Our first heuristic is an adaptation of the greedy weighted set cover algorithm
first described by Chvatal~\cite{Chvatal_1979}.
Similar to the binary LP, it uses the rectangle powerset of all base rectangles
of a given polygon $\Polygon$ as candidate set and seeks to ensure that all base
rectangles are covered.
The algorithm incrementally computes a cover $\Cover$ as follows:
Let $\Interior(\Cover) = \bigcup_{\Rect \in \Cover} \Interior(\Rect)$.
For each rectangle $\Rect \in \RectPow(\BaseRects(\Polygon))$,
the algorithm maintains its \emph{effective area} $a(\Rect)$ as
the area of $\Interior(\Rect) \setminus \Interior(\Cover)$.
Furthermore, the algorithm keeps track of all base rectangles that are
currently uncovered.
Note that as the rectangle powerset provides the candidates to
build $\Cover$ and two base rectangles can share at most a common edge but
never properly intersect, a base rectangle is either fully covered or fully
uncovered, but never partially covered.
Among all candidate rectangles, the algorithm chooses a rectangle $\Rect$
that minimizes the \emph{relative effective cost} $\frac{\Cost(\Rect)}{a(R)}$, adds it to $\Cover$,
and removes it from the set of candidates.
Afterwards, it updates the set of uncovered base rectangles,
as well as the effective area of all candidates, removing
any candidate with an effective area of $0$.
The procedure is repeated until no candidates are left.

\begin{lemma}\label{lem:greedy-time}
The Greedy Weighted Set Cover algorithm runs in $\bigO({|\BaseRects(\Polygon)|}^3) \subseteq \bigO(n^6)$ time, uses
$\bigO(|\BaseRects(\Polygon)|^2) \subseteq \bigO(n^4)$ space on a polygon $\Polygon$ with $2n$ edges, and returns an $\bigO(\log n)$-approximate solution for the \WRC{} problem.
\end{lemma}

\begin{proof}
Let $b := |\BaseRects(\Polygon)|$ be the number of base rectangles of $\Polygon$.
Computing the initial candidate set and initializing the respective effective
areas takes $\bigO(b^2)$ time by \autoref{cor:rectpow_size}.
In each iteration of the algorithm, choosing the candidate rectangle $\Rect$
takes $\bigO(b^2)$ time if we keep all candidates in a list%
\footnote{We could use a heap, but as we need to update all candidates later, this does not pay off.}
and search it linearly.
We keep the set of uncovered base rectangles in a hash table and assume that
updates and queries can be handled in (expected) constant time.
Obtaining the set of newly covered base rectangles takes time $\bigO(b)$.
Let $r_i$ be the number of newly covered base rectangles in iteration $i$.
To update the effective area of a candidate rectangle, we test whether it
covers any of the newly covered base rectangles in $\bigO(r_i)$ time
and, if necessary, update its relative effective cost in constant time.
As there are $\bigO(b^2)$ candidates, this takes $\bigO(b^2 \cdot r_i)$ in the $i$-th iteration.
Each iteration covers at least one base rectangle, thus there are at most $\bigO(b)$
iterations, yielding a total running time of
$\bigO(b \cdot b^2 + b \cdot b + b^2 \cdot \sum_i r_i) = \bigO(b^3)$,
as each base rectangle can only be newly covered once.

We need to maintain $\bigO(b^2)$ candidates and monitor $\bigO(b)$ base rectangles,
so the space complexity is $\bigO(b^2)$.
The $\bigO(n^6)$ running time and $\bigO(n^4)$ space follow by \autoref{lem:num_baserects}.
\end{proof}

\subsection{Strip Cover}\label{sect:greedy-strip} %
Our next algorithm is based on the strip cover algorithm by Kumar and
Ramesh~\cite{Kumar_Ramesh_1999} (cf.~\autoref{sect:related}).
In the original description, the algorithm uses grid rectangles and does not
state any running times.
Whereas the adaptation to the base rectangle graph is straightforward,
we also describe how to implement the algorithm efficiently:
We first obtain the base rectangle graph $\BaseRectGraph(\Polygon)$ of the input polygon $\Polygon$.
For each node $\BaseRect$ of $\BaseRectGraph(\Polygon)$, let $h(\BaseRect)$
denote the \emph{height} of $\BaseRect$, which is defined as the length of the longest
unidirectional ``bottom-going'' path that starts with the edge to its bottom neighbor, $\{\BaseRect,
b(\BaseRect)\}$.
Thus, if $\BaseRect$ has no bottom neighbor, $h(\BaseRect) = 0$.
Now, for each node $\BaseRect$ that does \emph{not} have a \emph{top} neighbor $t(\BaseRect)$,
follow the unidirectional ``left-going'' path starting with the edge $\{\BaseRect, \ell(\BaseRect)\}$
until a node $\BaseRect'$ is reached that either has no left neighbor or
$h(\ell(\BaseRect')) < h(\BaseRect)$.
Symmetrically, follow the unidirectional ``right-going'' path starting with the edge $\{\BaseRect, r(\BaseRect)\}$
until a node $\BaseRect''$ is reached that either has no right neighbor or $h(r(\BaseRect'')) < h(\BaseRect)$.
From $\BaseRect''$, follow the unidirectional ``bottom-going'' path starting with $\{\BaseRect'', b(\BaseRect'')\}$
and stop after $h(\BaseRect)$ steps.
Let $\BaseRect'''$ be the last node of this path.
Report the rectangle $\Rect$ with $\TopLeft(\Rect) = \TopLeft(\BaseRect')$ and $\BottomRight(\Rect) = \BottomRight(\BaseRect''')$.
The set of all rectangles $\Rect$ obtained in this way, eliminating duplicates,
is returned as cover.
Note that by construction, each rectangle $\Rect$ is maximal.

\begin{lemma}
The greedy strip cover algorithm runs in $\bigO(n^2)$ time,
uses $\bigO(|\BaseRects(\Polygon)|)$ space on a polygon $\Polygon$ with $2n$ edges,
and returns a cover of size $\bigO(n)$.
\end{lemma}
\begin{proof}
We first compute the base rectangle graph $\BaseRectGraph(\Polygon)$ in
$\bigO(n\log n + |\BaseRects(\Polygon)|)$ time
by \autoref{lem:compute_baserectgraph}.
Computing $h(\BaseRect)$ for each node $\BaseRect$ requires one traversal
of the graph and can thus be done in time linear in the size of $\BaseRectGraph(\Polygon)$.
Again by \autoref{lem:compute_baserectgraph}, the number of base rectangles
without top neighbor is in $\bigO(n)$ and each unidirectional path has length
$\bigO(n)$.
Hence, for each starting node $\BaseRect$, finding $\BaseRect'$, $\BaseRect''$,
$\BaseRect'''$ takes $\bigO(n)$ time.

The algorithm adds at most one rectangle to the cover for each base rectangle
without top neighbor, thus the size of the returned cover is in $\bigO(n)$.
Testing whether the resulting rectangle has already been found earlier
hence takes no more $\bigO(n)$ time, which yields a total running time of
$\bigO(n^2)$.%

Storing $\BaseRectGraph(\Polygon)$ requires $\bigO(|\BaseRects(\Polygon)|)$ space.
As we compute a cover of cardinality $\bigO(n)$, we need at most $\bigO(n)$
space to store the result and implement the duplicate elimination.
\end{proof}
As mentioned above, the cover computed by the algorithm is a subset of all
maximal rectangles $\MaxRects(\Polygon)$, but $|\MaxRects(\Polygon)|$ can be as
large as $\bigO(n^2)$~\cite{DBLP:journals/siamdm/Franzblau89}.

Note that the Strip Cover algorithm does not consider any cost
function.
To make it ``cost-aware'', we will combine it with postprocessors, which
are described in \autoref{sect:postproc}.
\subsection{Partition Algorithm}\label{sect:partition} %
For reasons of self-containedness, we here briefly review Ohtsuki's
optimal rectangle partitioning algorithm~\cite{Ohtsuki_1982} before we describe our adaptations for
the rectangle cover problem, which are implemented as postprocessing routines and described in \autoref{sect:join}.

The algorithm computes a partition by first locating %
pairs of concave vertices that can be connected by a straight horizontal or vertical segment in a given polygon $\Polygon$'s interior.
Such a segment is also called \emph{degenerate diagonal} or \emph{chord}.
Construct a graph with a node for each chord and an edge between two nodes if the corresponding chords properly %
cross each other in the interior.
As two horizontal (two vertical) chords never cross, the graph is bipartite.
Next, compute a maximum-cardinality independent set $I$, which can be done in polynomial time on bipartite graphs via matching. %
The chords in $I$ partition $\Polygon$ into smaller polygons.
For each concave vertex in the remaining, smaller polygons, extend a horizontal (or vertical)
line into the interior until it hits an edge of the polygon.
The result is an optimal partition of $\Polygon$ into a set of rectangles.

\begin{lemma}[\cite{Hiroshi_Takao_1986}]\label{lem:partition-time}
The partition algorithm can be implemented to run in time $\bigO(n^{1.5} \log{n})$
on a polygon with $2n$ edges.
\end{lemma}

\subsection{Postprocessing}\label{sect:postproc} %
Recall that neither the Strip Cover algorithm (\autoref{sect:greedy-strip}) nor
the Partition algorithm (\autoref{sect:partition}) use rectangle weights to
compute a cover.
We therefore introduce four postprocessing routines that connect the aforementioned
algorithms to the \WRC{} problem and also serve to improve the initial solution
by the Greedy Set Cover algorithm (\autoref{sect:greedy-set}).
All postprocessors assume a given initial cover $\Cover$ and can be chained.
\subsubsection{Prune}\label{sect:min}\label{sect:prune} %
This postprocessor tests for each rectangle $\Rect \in \Cover$ whether all base
rectangles contained in it are also covered by at least one other rectangle in
the cover.
If this is the case, rectangle $\Rect$ is removed from $\Cover$ and the
algorithm continues with the next rectangle.

\begin{lemma}\label{lem:prune}
The prune postprocessor runs in time
$\bigO(|\Cover|\cdot |\BaseRects(\Polygon)|)$
for a polygon $\Polygon$ and an initial cover $\Cover$.
\end{lemma}
\begin{proof}
We first compute the base rectangle graph $\BaseRectGraph(\Polygon)$ in
$\bigO(n\log n + |\BaseRects(\Polygon)|)$ time by \autoref{lem:compute_baserectgraph}.
We maintain a counter for each base rectangle $\BaseRect \in \BaseRects(\Polygon)$ that stores
the number of cover rectangles it is contained in.
To initialize the counters, for each rectangle $\Rect \in \Cover$, we iterate
over all base rectangles that are fully contained in $\Rect$ and increase their
counter.
We then iterate a second time over all $\Rect \in \Cover$, check each for
redundancy using the counters, and update the counters where necessary.
The total running time hence is $\bigO(|\Cover|\cdot |\BaseRects(\Polygon)|)$.
\end{proof}

\subsubsection{Trim}\label{sect:trim} %
Similar to the prune postprocessor, trim checks for redundancies, however
w.r.t.\ the area of a rectangle in the cover.
For each rectangle $\Rect \in \Cover$, trim identifies the set of base rectangles $U_{\Rect}$ that
are only contained in and thus covered by $\Rect$.
It then replaces $\Rect$ in $\Cover$ by the bounding box of these base
rectangles, $\BBox(U_{\Rect})$.
Trim can be implemented analogously to prune.

\begin{lemma}
The trim postprocessor runs in time
$\bigO(|\Cover|\cdot |\BaseRects(\Polygon)|)$
for a polygon $\Polygon$ and an initial cover $\Cover$.
\end{lemma}

\subsubsection{Rectangle Splits}\label{sect:split} %
The idea behind rectangle splits is to eliminate rectangles that have large
overlap with other rectangles in the cover, but cannot be pruned or trimmed any
further.
To split a rectangle $\Rect$, we first remove it from the cover.
As $\Rect$ could not be trimmed any further, this results in a non-empty set of
uncovered base rectangles $U_{\Rect}$.
For each polygon formed by a maximally connected subset of $U_{\Rect}$, called
\emph{gap}, the rectangle split postprocessors try to newly cover the gap
using different approaches.
The split is accepted if the total cost of the new cover of $U_{\Rect}$ is less than
the cost of $\Rect$, and rejected and undone otherwise.

\begin{lemma}\label{lem:gap-vertices}
For a set of gaps formed by the uniquely covered base rectangles $U_{\Rect}$ of
a rectangle $\Rect$ within a cover $\Cover$ for a polygon $\Polygon$, the total
number of vertices of all gaps is in $\bigO({|\Cover|}^2)$.
\end{lemma}
\begin{proof}
The polygon $Q$ resulting from the removal of $\Rect$ from $\Cover$ is the
union of $|\Cover|-1$ rectangles.
As the union complexity of a set of rectangles is quadratic in the size of
the set~\cite{DBLP:journals/combinatorics/KellerS18},
both $\Polygon$ and $Q$ have $\bigO({|\Cover|}^2)$ vertices.

Each vertex of a gap is either a vertex of $Q$, $\Rect$, or lies
on the intersection of an edge of $\Rect$ and $Q$ each, making
it a vertex of $\Polygon$.
Hence, the total number of vertices of all gaps is $\bigO({|\Cover|}^2)$.
\end{proof}

\paragraph{Bounding Box Split.}
This postprocessor simply covers each gap by its bounding box.
As each gap originally was contained in a larger rectangle that was part of the cover,
the bounding box of a gap must be fully contained in the polygon's interior.

\begin{lemma}\label{lem:bbox-split}
The bounding box split postprocessor runs in time $\bigO({|\Cover|}^3)$
for a polygon $\Polygon$ and a cover $\Cover$.
\end{lemma}
\begin{proof}
By \autoref{lem:gap-vertices}, the total number of vertices for all gaps resulting
from the removal of a single rectangle from the cover is in $\bigO({|\Cover|}^2)$.
Hence, for each rectangle, there are at most $\bigO({|\Cover|}^2)$ gaps to
cover, resulting in a total running time of $\bigO({|\Cover|}^3)$.
\end{proof}

\paragraph{Partition Split.}
As the name suggests, this postprocessor covers each gap using the
partition algorithm (cf.\ \autoref{sect:partition}). 

\begin{lemma}\label{lem:partition-split}
The partition split postprocessor runs in time $\bigO({|\Cover|}^4\log|\Cover|)$
for a polygon $\Polygon$ and a cover $\Cover$.
\end{lemma}
\begin{proof}
By \autoref{lem:gap-vertices}, the total number
of vertices for all gaps
resulting from the removal of a single rectangle from the cover is in
$\bigO({|\Cover|}^2)$.
Hence, computing a partition for all gaps can be done in
$\bigO({|\Cover|}^3\log|\Cover|)$ time per removed rectangle by \autoref{lem:partition-time},
and thus in $\bigO({|\Cover|}^4\log|\Cover|)$ time overall.
\end{proof}

\subsubsection{Rectangle Joins}\label{sect:join} %
We now describe two postprocessors that try to replace a set of rectangles
contained in the cover by a single rectangle if this improves the cost of the
cover.
The postprocessors differ in the set of rectangles they consider for a join.
Joining a set of rectangles $\Rectangles \subseteq \Cover$
removes $\Rectangles$ from $\Cover$ and inserts the bounding box of the set $\BBox(\Rectangles)$.
\paragraph{Simple Join.}
For the simple or aligned join, the algorithm only considers sets of rectangles where all
rectangles have the same minimum and maximum $y$-coordinate, \ie, they are
horizontally aligned, or sets of rectangles where all rectangles have the same
minimum and maximum $x$-coordinate (vertically aligned).

\begin{lemma}
The simple aligned join postprocessor runs in time $\bigO(|\Cover|\cdot n)$
for a polygon $\Polygon$ with $2n$ edges and a cover $\Cover$.
\end{lemma}
\begin{proof}
We first iterate over all $\Rect \in \Cover$ and sort each into two buckets
identified by a pair of $x$- or $y$-coordinates, respectively, according to the
minimum and maximum $x$- or $y$-coordinates of $\Rect$.
We assume that the rectangles in a bucket do not overlap, \ie, if necessary
$\Cover$ has been pruned (cf.\ \autoref{sect:prune}) and trimmed (cf.\
\autoref{sect:trim}) before.
Thus, each bucket contains at most $\bigO(n)$ rectangles by
\autoref{lem:compute_baserectgraph} (longest unidirectional path).
We sort each bucket identified by $y$-coordinates according to the $x$-coordinates of its
rectangles, and vice-versa.
The test for joins can then be implemented as a sweep line process, where a rectangle $\Rect$
is joined with its predecessor rectangle $\Rect'$ if and only if the join leads to a reduction
in cost and the joined rectangle is fully contained in $\Interior(\Polygon)$ by
checking whether a straight diagonal line from $\TopLeft(\Rect')$ to
$\BottomRight(\Rect)$ intersects with an edge of the polygon in $\bigO(n)$ time.
The sweep line processes take $\bigO(|\Cover| \cdot n)$ total time and dominate
the rest.
\end{proof}
\paragraph{Full Join.}
The postprocessor works similar to the simple aligned join, but considers
joins of arbitrary sets of rectangles.

\begin{lemma}\label{lem:join-full}
The full join postprocessor runs in time $\bigO({|\Cover|}^2\cdot n)$
for a polygon $\Polygon$ with $2n$ edges and a cover $\Cover$.
\end{lemma}
\begin{proof}
Fix an arbitrary order of the rectangles in $\Cover$.
For each $\Rect \in \Cover$, test whether a join with any of its
successors improves the cost of the cover and yields a rectangle that is fully
contained in $\Interior(\Polygon)$.
If this is the case, the joined rectangle takes the place of $\Rect$ in the
order and the algorithm tries to join it with further successors.
Checking whether the joined rectangle is contained in $\Interior(\Polygon)$ can
be accomplished in $\bigO(n)$ time by testing for intersections with the polygon's edges.
This can be necessary at most $\bigO({|\Cover|}^2)$ times.
\end{proof}

\section{Experiments}\label{sect:experiments}%
We implemented\footnote{\url{https://github.com/WeRecCover/WeRecCover}}
all algorithms described in \autoref{sect:algorithms} in
\Cpp{}17
and compare them against each other on a large and diverse set of instances,
having \num{186837} polygons in total.
To the best of our knowledge, this is the first experimental evaluation of
algorithms for the \WRC{} problem.
For the smaller instances, we include exact solutions based on the ILP
from
\autoref{sect:bip}.
It would have been interesting to compare our algorithms also to
the greedy set cover-based algorithm by Heinrich-Litan and Lübbecke~\cite{Heinrich-Litan_Lubbecke_2007} for the unweighted case.
Unfortunately, the code is no longer available.
\subsection{Methodology} %
We compiled our code with GCC 11.4\footnote{with \texttt{-O3 -march=native -mtune=native}}
and used Gurobi 10.02\footnote{\url{https://gurobi.com}} as ILP solver.
To load instances and for geometric computations, we rely on the
Computational Geometry Algorithms Library~\cite{cgal:eb-23b}. %
All our implementations are deterministic and single-threaded, but Gurobi is
inherently multi-threaded.
For this reason, we do not compare the running times of the ILP-based algorithm
with the others, but only the solution quality.
All runtime experiments were conducted on a machine
with an Intel(R) Xeon(R) E5 %
CPU and \SI{1.5}{\tera\byte} of main memory, running Ubuntu Linux 22.04 (Kernel
5.15).
To counteract errors of measurement, every experiment was assigned to an
exclusive core and carried out three times.
For each algorithm, polygon, and cost function, we report the \emph{absolute solution
quality}, \ie, the total cost of the cover, and median of the \emph{absolute running time}.
Furthermore, for each such triple, we computed the \emph{relative solution quality} and
and \emph{relative running time} as the ratio of the corresponding absolute value
and the best value obtained by any algorithm for this polygon and cost function.
We set a timeout of \SI{1}{\hour} per polygon for the ILP solver,
and \SI{4}{\hour} per instance for the other algorithms.

\subsection{Algorithms, Instances, and Cost Functions} %
\begin{table}[tb]
\small
\centering
\caption{Algorithms evaluated experimentally.}\label{tab:algorithms}
\renewcommand{\arraystretch}{0.8}
\begin{tabular}{@{}lcc@{}}
\toprule
Base Algorithm & Postprocessors & Short Name \\
\midrule
\AlgName{partition} & -- & \AlgName{par} \\
 & \AlgName{simple join} & \AlgName{par-j} \\
 & \AlgName{full join} & \AlgName{par-f} \\
\midrule
\AlgName{strip} & -- & \AlgName{strip} \\
 & \AlgName{prune}, \AlgName{trim} & \AlgName{strip-pt} \\
 & \AlgName{prune}, \AlgName{trim}, \AlgName{bb-split} & \AlgName{strip-ptb} \\
 & \AlgName{prune}, \AlgName{trim}, \AlgName{par-split} & \AlgName{strip-pts} \\
\midrule
\AlgName{greedy} & -- & \AlgName{grdy} \\
 & \AlgName{prune}, \AlgName{trim} & \AlgName{grdy-pt} \\
\midrule
\AlgName{ilp} & -- & \AlgName{ilp} \\
\bottomrule
\end{tabular}
\end{table}
\begin{table}[tb]
\small
\centering
\caption{Instances used in the evaluation, where
\#$\Polygon$: \#non-trivial polygons,
$\max\Holes$/$\max\AllVertices$/$\max\GridRects$/$\max\BaseRects$: max.\ \#holes/vertices/grid rectangles/base rectangles.} %
\label{tab:instances}
\begin{tabular}{@{}lrrrrrr@{}}
\toprule
data set               & \#$\Polygon$ & $\max\Holes$ & $\max\AllVertices$ & $\max\GridRects$ & $\max\BaseRects$ & $\max\frac{\GridRects}{\BaseRects}$\\
\midrule
\InstanceSet{ccitt}    &  \num{11970} & \num{192}      & \num{27478}        & \num{539347} & \num{159132} & \num{3.39}\\
\InstanceSet{icons}    &  \num{211}   & \num{10}      & \num{108}         & \num{134} & \num{134} & \num{1.86}\\
\InstanceSet{nasa}     &  \num{13230} & \num{9804}     & \num{124204}       & \num{1826606} & \num{1208309} & \num{3.43}\\
\InstanceSet{caltech}  &  \num{22989} & \num{1929}    & \num{10986}       & \num{55379} & \num{42982} & \num{3.69}\\
\InstanceSet{textures} &  \num{2674}  & \num{16249}    & \num{160132}       & \num{825464} & \num{713421} & \num{2.94}\\
\InstanceSet{aerials}  &  \num{131818}& \num{24079}    & \num{242330}       & \num{583977} & \num{879458} & \num{1.76}\\
\InstanceSet{dats}     &  \num{3919}  & \num{14530}    & \num{69740}       & \num{250239} & \num{235841} & \num{3.56}\\
\InstanceSet{cgshop}   &  \num{26}    & \num{67}      & \num{94778}       & \num{26293767} & \num{988050} & \num{145.84}\\
\bottomrule
\end{tabular}
\end{table}
\paragraph{Algorithms.}
We included all algorithms described in \autoref{sect:algorithms} and combined them with
different postprocessing routines.
\autoref{tab:algorithms} gives an overview of the different combinations we
evaluated and the shorthand names we use in the discussion of the results.

\paragraph{Instances.}
We used a diverse collection of instances from experimental evaluations of
related problems~\cite{Heinrich-Litan_Lubbecke_2007,Koch_Marenco_2022}.
Where necessary, we converted the instances to the WKT\footnote{Well-known
text representation of geometry} format first. With the exception of the \InstanceSet{dats} and \InstanceSet{cgshop} instance sets, all sets of instances were retrieved from Koch and Marenco's~\cite{Koch_Marenco_2022} repository of images they binarized\footnote{\label{imgsrc:repo}\url{https://drive.google.com/drive/folders/1EPj1w_P8Bgg_86dCzOWJVu3JnFsEbrPO}~\cite{Koch_Marenco_2022}}.
Even though our algorithms can also handle real-valued polygons, all instances
here only use integer coordinates.
Our instance sets are:
\\
\InstanceSet{ccitt}: fax test images, originally provided to compare compression
  algorithms;
\InstanceSet{icons}: small black and white icons ranging from 6x6 to 24x24 pixels;
\InstanceSet{nasa}: binary versions of large satellite images provided by NASA;
\InstanceSet{caltech}: binary versions of various standard, medium-sized images used in computer vision research;
\InstanceSet{textures}: binary versions of texture images used in image processing research;
\InstanceSet{aerials}: binary versions of aerial photographs used in image processing research;
\InstanceSet{dats}: black and white images introduced %
in~\cite{Heinrich-Litan_Lubbecke_2007};
\InstanceSet{cgshop}: rectilinear instances from the \mbox{CG:SHOP} 2023 competition\footnote{\url{https://cgshop.ibr.cs.tu-bs.de/competition/cg-shop-2023/}}.

Each instance in these collections is a ``multipolygon``, \ie, it can consist
of multiple polygons.
Our algorithms solve each polygon individually and independently.
In many cases, hole-free rectangles are among these polygons, which can be
covered optimally by a single rectangle.
This case is easy to recognize and handle by our preprocessor, as it does not
provide any insights into the performance of the main algorithm.
We therefore do not consider these trivial polygons in our evaluation any further.
To facilitate an evaluation w.r.t.\ the metrics of a polygon, we report
solution quality and running time for each non-trivial polygon in an instance
rather than the sum over all polygons in the multipolygon.
See \autoref{tab:instances} for a summary of statistics, 
including the maximum number of grid and base rectangles for each instance set
as well as the maximum ratio between them.
The average ratio of grid and base rectangles was between \num{1.01} (\InstanceSet{textures}) and \num{11.24} (\InstanceSet{cgshop}).
The maximum size of the powerset of base rectangles, which is used by \AlgName{grdy} and \AlgName{ilp}, was between \num{1793} (\InstanceSet{icons}) and \num{7901671022} (\InstanceSet{nasa}).
In total, our study comprises \num{186837} non-trivial polygons.

\paragraph{Cost Functions.}
As the different tradeoffs between rectangle creation cost ($\alpha$) and rectangle area cost ($\beta$) mainly
depend on the ratio of these two values, we only vary $\alpha$ in our experiments and keep $\beta = 1$ fixed.
For all instances, we ran experiments with $\alpha \in \{1,10,50,100,500,1000\}$.
Recall from \autoref{lem:min_baserect} that if $\alpha \leq \beta \cdot A_{\min}$, where $A_{\min}$
is the smallest area of any base rectangle, the partition algorithm is optimal.
As our instances are integer-valued, the partition algorithm is guaranteed optimal for $\alpha \leq 1$.

\subsection{Results} %
We only ran algorithm \AlgName{ilp} on polygons from the instance sets with smaller polygons
\InstanceSet{icons}, \InstanceSet{dats}, and \InstanceSet{cgshop}.
Still, \AlgName{ilp} ran into timeouts or was terminated due to insufficient
memory on polygons from two \InstanceSet{dats} instances and five
\InstanceSet{cgshop} instances.
Similarly, algorithm \AlgName{grdy} ran into several timeouts and memory
allocation failures on polygons from \InstanceSet{aerials}, \InstanceSet{nasa},
\InstanceSet{textures}, and \InstanceSet{dats}, and \InstanceSet{cgshop}.
We therefore do not compare \AlgName{grdy} on polygons from these instances.
The two instances with the largest number of base rectangles on which \AlgName{grdy}
terminated are from the \InstanceSet{aerials} instance set and have \num{988050} base rectangles (size of powerset: \num{1.53}~billion).
All instance on which \AlgName{grdy} terminated have at most \num{108970} vertices.
\AlgName{par-f} ran into a timeout for one instance of \InstanceSet{aerials}
and one of \InstanceSet{textures}, where $\alpha \geq 500$.
As these were the only timeouts, we let it finish nonetheless in order to
have a complete set of results for the evaluation.
\paragraph{Solution Quality.}
\begin{figure}[tb]
\centering
\includegraphics[width=.49\linewidth]{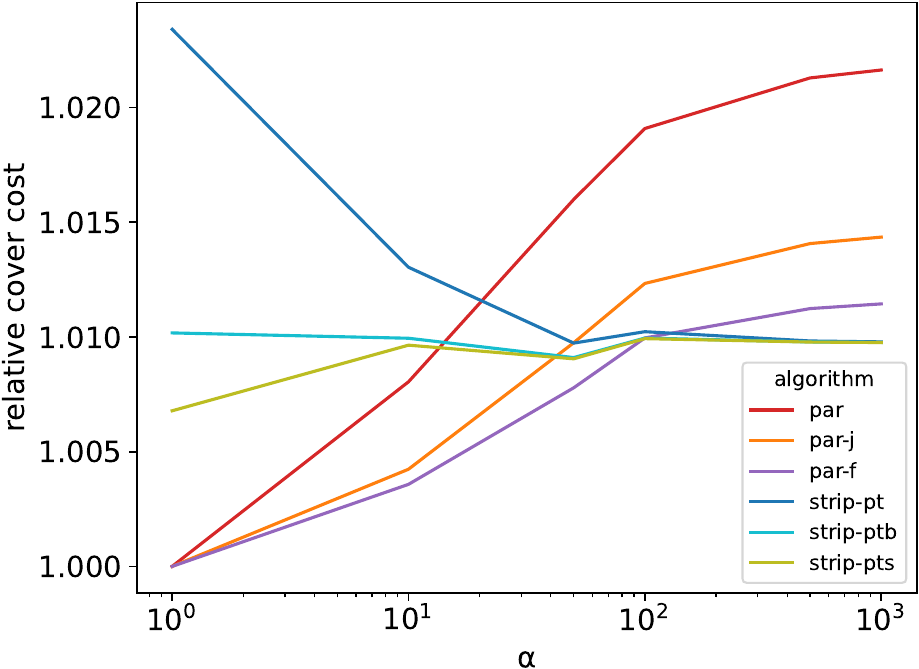}
\includegraphics[width=.49\linewidth]{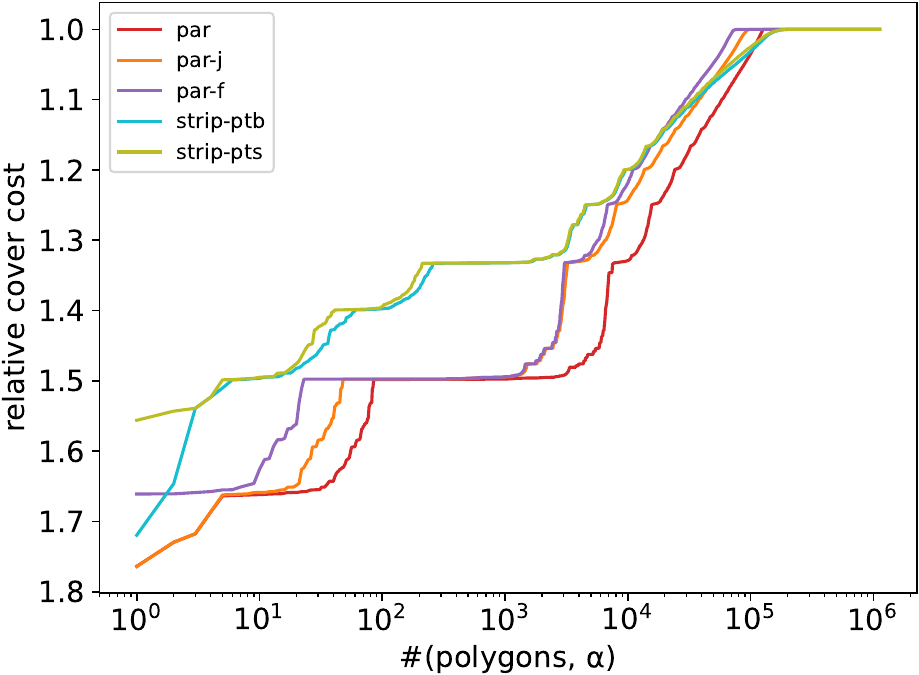}
\caption{%
Mean relative cost for different values of $\alpha$ across all instances (left),
and a performance plot showing on how many polygon/$\alpha$ pairs each algorithm returned a solution with at least a given relative cover cost (right).
In both cases, note the logarithmic x axis.
}%
\label{fig:alpha-vs-relcost}
\end{figure}

Across all instances, \AlgName{par-f} gave on average the best solution
quality for $\alpha \leq 100$,
whereas for $\alpha \geq 100$, \AlgName{strip-ptb} and \AlgName{strip-pts} exceeded
all others, see also \autoref{fig:alpha-vs-relcost} (left).
On polygons from \InstanceSet{caltech}, \InstanceSet{cgshop}, and \InstanceSet{dats},
\AlgName{par-f} yielded on average the best cover for all values of $\alpha$,
and was only slightly inferior to \AlgName{strip-pts} on \InstanceSet{aerials}
for $\alpha > 500$.
On \InstanceSet{ccitt}, \InstanceSet{icons}, \InstanceSet{nasa}, and
\InstanceSet{textures}, \AlgName{strip-pts} outperformed \AlgName{par-f}
already for $\alpha \geq 100$ or less (cf.\ \autoref{fig:alpha-vs-relcost-appdx}).

\begin{figure}[tb]
\centering
\includegraphics[width=.49\linewidth]{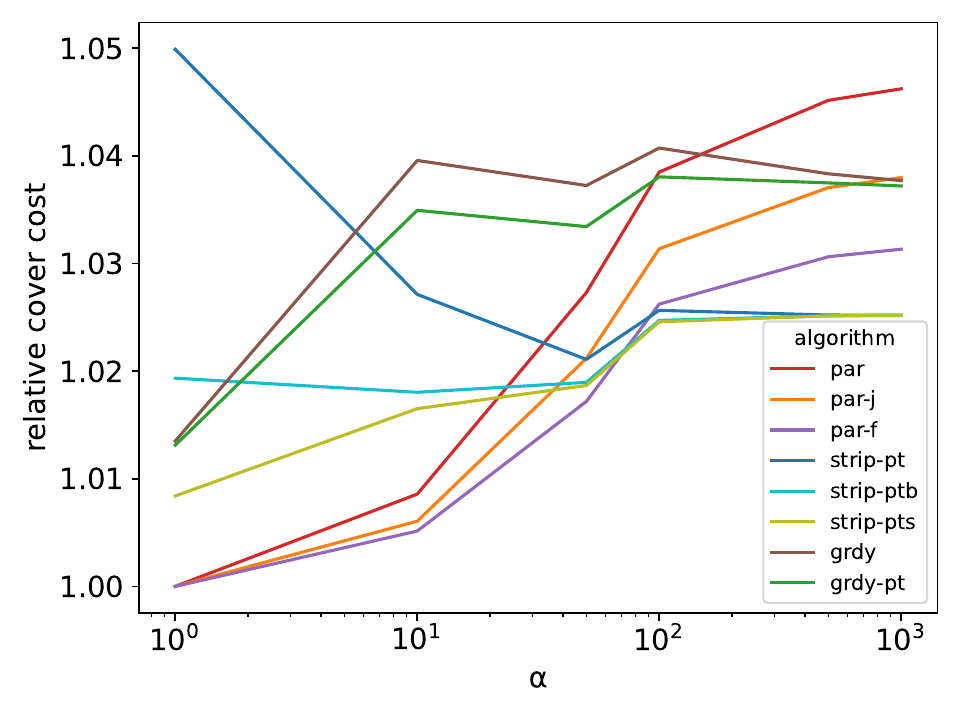}
\includegraphics[width=.49\linewidth]{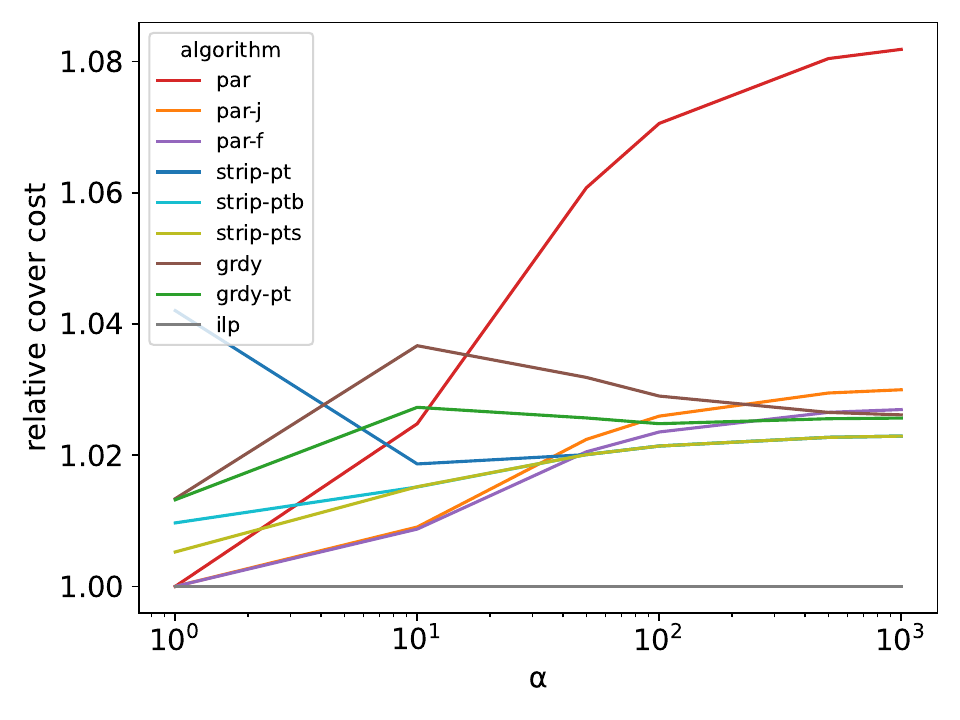}
\caption{%
Mean relative cost for different values of $\alpha$, across all instances where \AlgName{grdy} terminated, \ie, \InstanceSet{icons}, \InstanceSet{ccitt}, and \InstanceSet{caltech} (left)
and only on \InstanceSet{icons} (right).
In both cases, note the logarithmic x axis.
}%
\label{fig:alpha-vs-relcost-appdx}
\end{figure}

The performance of \AlgName{strip-pt}, \AlgName{strip-ptb}, and \AlgName{strip-pts} is similar
for large values of $\alpha$, whereas \AlgName{strip-pt} and \AlgName{strip-ptb} are
clearly inferior to \AlgName{strip-pts} for smaller values.
The solution quality was significantly impaired if \AlgName{strip} was used
without \AlgName{trim} and yielded on average up to \SI{44}{\percent} and
\SI{23}{\percent} heavier covers for small $\alpha$.
As $\alpha$ grows, the gap in average relative solution cost widened for
\AlgName{par}, and \AlgName{par-j}, and \AlgName{par-f}, showing that
the postprocessors gain in effectiveness for large $\alpha$.
Recall that \AlgName{strip} and \AlgName{par} without postprocessors do not
consider the given cost function, hence their inferior performance is not
surprising.
On instance sets where \AlgName{grdy} yielded a  result for all polygons, \ie,
\InstanceSet{icons}, \InstanceSet{caltech}, and \InstanceSet{ccitt}, both
\AlgName{grdy} and \AlgName{grdy-pt} were not significantly better on average
than \AlgName{par-f} and \AlgName{strip-pts} and partially also \AlgName{par-j}
for any value of $\alpha$, and only better than \AlgName{strip-ptb} for $\alpha
= 1$ (cf.\ \autoref{fig:alpha-vs-relcost-appdx}).

Looking at the maximum relative cost rather than the average, \AlgName{par-j},
\AlgName{par-f} and \AlgName{strip-pts} never exceeded the best solution by
more than \SI{76}{\percent}, \SI{66}{\percent}, and \SI{56}{\percent},
respectively (cf.\ \autoref{fig:alpha-vs-relcost}), and found the best cover
in around \SI{90}{\percent} of all cases.
By contrast, \AlgName{strip} and \AlgName{strip-pt} produced solutions that are
by a factor of \num{80} and \num{4.7}, respectively, larger than the best.
For $\alpha \geq 50$, \AlgName{strip-pts} always had the smallest maximum
relative cost when looking at all instance sets.
The picture is also similar on each instance set individually, except for \InstanceSet{icons},
where both \AlgName{par-j} and \AlgName{par-f} outperformed \AlgName{strip-pts} for all values of $\alpha$
and never returned a cover with weight \SI{33}{\percent} larger than the optimal.

\emph{In summary, if worst-case performance is prioritized, \AlgName{strip-pts}
is the candidate of choice for $\alpha \geq 50$, whereas where average performance
is considered, \AlgName{par-f} is superior for small values of $\alpha$,
followed by \AlgName{par-j}.
Furthermore, all postprocessing routines proved themselves very effective in increasing
the solution quality.}

\paragraph{Running Time.}
The fastest algorithm across all cost functions and instance sets was
\AlgName{par} with mean relative running times of around \num{1.00},
very closely followed by \AlgName{par-j} (around \num{1.06}).
\AlgName{par-f} was only slightly slower than \AlgName{par-j} except for
\InstanceSet{cgshop}, where its mean relative running time was around \num{2.65}.
\AlgName{strip}, \AlgName{strip-p}, and \AlgName{strip-pt} on average had a relative
running time of \numrange{7.7}{8.0}, whereas \AlgName{strip-ptb} and \AlgName{strip-pts}
were on average by a factor of \num{10} and \num{11} slower than \AlgName{par}.
On those instances where \AlgName{grdy} and \AlgName{grdy-pt} completed, they
had relative running times of around \num{7.3} (\InstanceSet{caltech}),
\num{3.4} (\InstanceSet{ccitt}), and \num{1.7} (\InstanceSet{icons}).
Similar to \AlgName{strip-pt}, the postprocessing of \AlgName{grdy-pt} only led
to a small increase in running time in comparison to \AlgName{grdy}.
On \InstanceSet{ccitt} and \InstanceSet{icons}, \AlgName{grdy} and \AlgName{grdy-pt}
were faster on average than \AlgName{strip-ptb} and \AlgName{strip-pts}.

Only \AlgName{par-f} showed a strong dependency between running time and $\alpha$,
which generally increased together with $\alpha$ on all instance sets.
The strongest increase was observed on \InstanceSet{cgshop}, where
the maximum running time for a polygon was \SI{239}{\second} for $\alpha=1$ and \SI{1618}{\second}
for $\alpha=\num{1000}$.
The mean running time rose from \SI{13}{\second} to \SI{83}{\second}.
The observed increase is due to the implementation of the full join postprocessor:
When considering a join of two rectangles, it first checks whether this
would improve the solution, and only then checks whether the joined
rectangle actually lies within the polygon's interior.
The larger the creation cost $\alpha$, the more often
a join reduces the weight of the cover and the second check becomes necessary.

All algorithms finished within \SI{3}{\milli\second} on \InstanceSet{icons}.
Across all polygons, the maximum running time of the \AlgName{par}- and
\AlgName{strip}-based algorithms, except \AlgName{par-f}, was between
\SI{51}{\minute} and \SI{58}{\minute}.
\begin{figure}[tb]
\centering
\includegraphics[width=.49\linewidth]{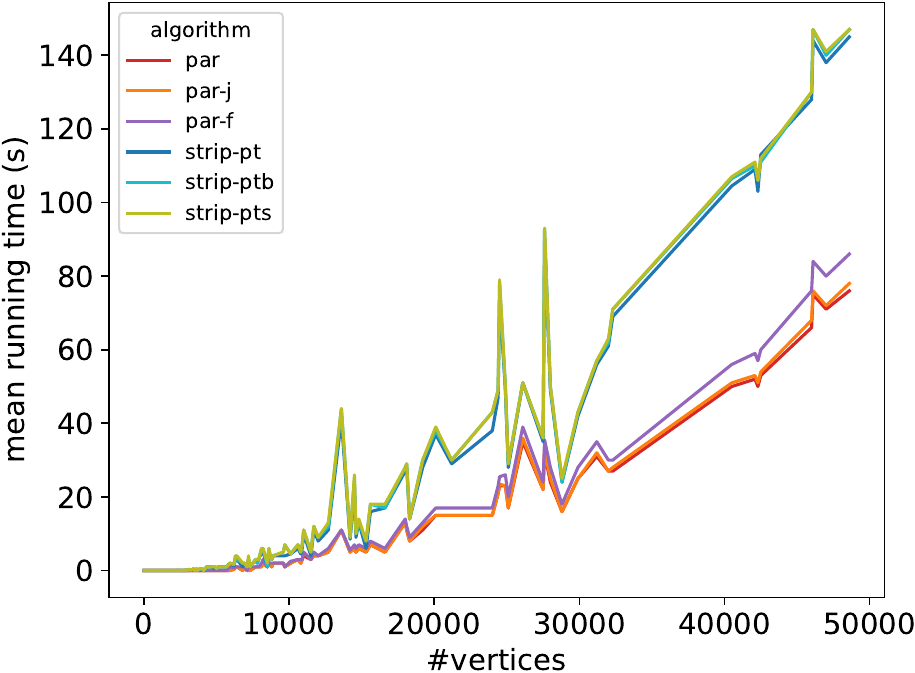}
\includegraphics[width=.49\linewidth]{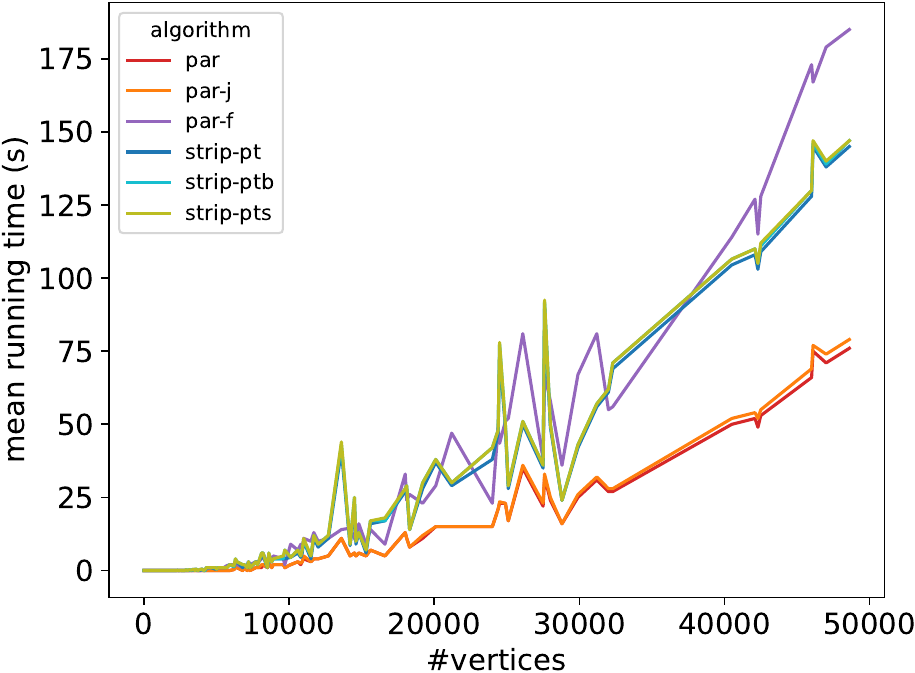}
\caption{%
Running time vs.\ number of vertices in the polygon for $\alpha = 10$ (left)
and $\alpha=100$ (right).
}%
\label{fig:time-byN}
\end{figure}
To asses the empirical running time of the \AlgName{par}- and
\AlgName{strip}-based algorithms further, we grouped the polygons by their
number of vertices in steps of \num{100} and computed the arithmetic mean
running time for every algorithm and number of vertices group.
The results, depicted in \autoref{fig:time-byN} for
$\alpha = 10$ and $\alpha = 100$, show a super-linear dependency of the running time on the
number of vertices also in practice and again demonstrate the runtime dependency of
\AlgName{par-f} on $\alpha$ in practice.

\emph{In summary, \AlgName{par-f} is the algorithm of choice for small values of $\alpha$.
It is very fast and offered the best solution quality in this setting.
For larger values of $\alpha$, \AlgName{strip-pts} is well-suited if solution quality is
important and its higher running time is acceptable, which should in particular be the
case for polygons with fewer vertices.
If both $\alpha$ and the polygon are large, our experiments suggest \AlgName{par-j} as a good compromise
between running time and solution quality.
}
\section{Conclusion and Future Work}\label{sect:conclusion}%
In this paper, we initiated the study of the \textsc{Weighted Rectilinear Cover} problem.
We introduced the concept of base rectangles and demonstrated that our new heuristics are both fast and deliver \emph{close-to-optimal} results.
In particular, they are faster and better than the greedy set cover approximation algorithm.
We focused on rectangle cost functions that take a rather specific, yet practical form, and were mainly interested in the solution quality that can be achieved in practice.

A logical next step would be to look into more general cost functions.
As long as an optimal solution can be built from base rectangles, our algorithms can be directly applied. %
Another direction is to focus on the time and memory requirements of algorithms that achieve similar solution costs.
A more theoretical question is, if improved, i.e. sub-logarithmic, approximation bounds for the \WRC{} problem are possible.

\clearpage
\bibliographystyle{plain}  %
\bibliography{paper}

\clearpage

\end{document}